\documentclass[journal,comsoc]{IEEEtran}
\usepackage[T1]{fontenc}

\ifCLASSINFOpdf
\else
\fi
\usepackage{cite}
\usepackage{mathtools}
\usepackage{algorithm2e}
\usepackage{amsthm}
\newtheorem{theorem}{Theorem}

\begin{document}
%
\title{Learning End-User Behavior for Optimized Bidding in HetNets: Impact on User/Network Association}
%
%
%

\author{Mohammad~Yousefvand,~\IEEEmembership{Member,~IEEE,} Narayan~Mandayam,~\IEEEmembership{Fellow,~IEEE,}
        
\thanks{M. Yousefvand and N. Mandayam are with the Wireless Information Network Lab (WINLAB), Department of Electrical and Computer Engineering, Rutgers University, New Brunswick
	NJ, 08902,  e-mails: {\{my342, narayan\}@winlab.rutgers.edu}.}
\thanks{This work is supported in part by the U.S. National Science Foundation under Grant No. 1421961 and Grant No. ACI-1541069}
}

%
%

\markboth{Journal of \LaTeX\ Class Files,~Vol.~14, No.~8, August~2019}%
{Shell \MakeLowercase{\textit{et al.}}: Bare Demo of IEEEtran.cls for IEEE Communications Society Journals}
%



\maketitle

\begin{abstract}
We study the impact of end-user behavior on service provider (SP) bidding and user/network association in a HetNet with multiple SPs while considering the uncertainty in the service guarantees offered by the SPs. Using Prospect Theory (PT) to model end-user decision making that deviates from expected utility theory (EUT), we formulate user association with SPs as a multiple leader Stackelberg game where each SP offers a bid to each user that includes a data rate with a certain probabilistic service guarantee and at a given price, while the user chooses the best offer among multiple such bids. We show that when users underweight the advertised service guarantees of the SPs (a behavior observed under uncertainty), the rejection rate of the bids increases dramatically which in turn decreases the SPs utilities and service rates. To overcome this, we propose a two-stage learning-based optimized bidding framework for SPs.  In the first stage, we use a support vector machine (SVM) learning algorithm to predict users' binary decisions (accept/reject bids), and then in the second stage we cast the utility-optimized bidding problem as a Markov Decision Problem (MDP) and propose a reinforcement learning-based dynamic programming algorithm to efficiently solve it. Simulation results and computational complexity analysis validate the efficiency of our proposed model.
\end{abstract}


%
\IEEEpeerreviewmaketitle

\section{Introduction}
%
%
%

\IEEEPARstart{U}{ltra-dense} heterogeneous networks (HetNets) with 
overlaid macro and small cells is a design solution that has emerged in 5G networks to cope with the increasing demand from mobile users for more capacity and higher data rates. In such HetNets users could have the option of choosing from multiple service providers (SPs), the one that offers them the best data service to meet their current demands and application requirements. In such a scenario, SPs will have to dynamically compete with each other to get the users to connect to their network by offering attractive service guarantees and prices. Therefore, user/network association and smart data pricing are extremely important from both user and SP perspectives \cite{Liu:2016:UAN,Sen:2013:SDP,Sen:2015:SDP}.
Earlier work in the literature on both user/network association as well as pricing (discussed in section II) has primarily relied on models based on Expected Utility Theory (EUT). When a SP controls access to end-users via differentiated and hierarchical monetary pricing, then the performance of the network is directly subject to end-user decision-making that has shown to deviate from EUT in many cases and is better captured by models based on the Nobel Prize-winning Prospect Theory (PT) \cite{Kahneman:1979:PT}. 

Due to uncertainties in channel and traffic conditions, when the advertised data rates offered by SPs can only be met with probabilistic guarantees, then these probabilities are not necessarily the same from user and SPs perspectives owing to the subjective biases of human decision making. This disparity can result in the rejection of SP offers as shown in our previous work \cite{Yousefvand:2018:IEU}. In fact, it was shown in \cite{Yousefvand:2018:IEU} that when users underweight the advertised service guarantees of the SPs (a behavior observed under uncertainty), the rejection rate of the bids increases dramatically which in turn decreases the SPs utilities and service rates. In this paper, to overcome this, we propose a two-stage learning-based optimized bidding framework for SPs.  In the first stage, we use a support vector machine learning algorithm to predict users' binary decisions (accept/reject bids), and then in the second stage we cast the utility-optimized bidding problem as a Markov Decision Problem (MDP) and propose a reinforcement learning-based dynamic programming algorithm to efficiently solve it. Simulation results and computational complexity analysis validate the efficiency of our proposed model.

The rest of this paper is organized as follows. In section \ref{sec:RelatedWork}, we review the related works. In section \ref{sec:SystemModel}, we present the HetNet model used along with optimization problems from the user and SP perspectives. In section \ref{sec:LBOB}, we introduce the proposed two-stage learning-based optimized bidding method for SPs to learn the end users behavior and optimize their bids based on that. We validate the efficiency of the proposed model in section \ref{sec:SimulationResults}, and conclude in section \ref{sec:Conclusion}.

\section{Related Work}
\label{sec:RelatedWork}
Numerous studies have been done in recent years to address user association in HetNets \cite{Liu:2016:UAN}. Due to the inherent interdependencies between user/cell association and resource allocation problems in HetNets and their direct impacts on each other, these two problems are usually jointly addressed and optimized with respect to a given performance parameter \cite{Fooladivanda:2013:JRA,Kuang:2016:OJU}. 
The diverse range of performance criteria considered in the user/network association schemes proposed for HetNets spans from load balancing \cite{Hirata:2018:OHD,Huang:2018:LOC,Lee:2018:UAB}, to energy efficiency \cite{Ding:2018:EEU,Li:2018:EEU,Farooq:2018:UTP,Yin:2019:EAJ}, interference management \cite{Feng:2018:IMU}, coverage maximization \cite{Hattab:2018:CRM}, and latency minimization \cite{Elkourdi:2018:TLL}. Also, in terms of the methodology, the range of approaches proposed for user/network association in HetNets spans from optimization methods \cite{Li:2018:EEU}, to game theoretic solutions \cite{Yin:2019:EAJ,Elkourdi:2018:TLL,Sun:2018:CAW,Brihi:2018:NPQ,Waheidi:2019:UDM,Haddad:2017:HNS}, evolutionary algorithms\cite{Amine:2018:NUA}, estimation methods \cite{Hattab:2018:LTR} and learning algorithms \cite{Pervez:2018:MBU}.

Moreover, several works investigated the role of network economics on user/network association and allocation of network resources, via supply and demand-based smart pricing models \cite{Iosifidis:2015:DAM,Im:2016:AMUSE, Ma:2018:DPP,Wang:2019:DCM,Sen:2019:TDP}. For example, in \cite{Iosifidis:2015:DAM} an iterative double-auction mechanism is proposed to offload traffic from macro BSs to third party small cell access points. A cost-aware adaptive bandwidth management mechanism for empowering users to make traffic offloading decisions is proposed in \cite{Im:2016:AMUSE}. The effects of users social learning on service provider’s dynamic pricing policies is investigated in \cite{Ma:2018:DPP}. In \cite{Wang:2019:DCM} the duopoly competition between mobile network operators for mobile data plans with time flexibility is characterized using a three stage game model. In \cite{Sen:2019:TDP} a framework for realization of time depending pricing for multimedia data traffic is proposed to modify users' behavior and prevent congestion. The role of PT in wireless data pricing has been explored in \cite{Li:2014:WUI,Yang:2014:IEU} where users use their subjective biases in evaluating objective probabilistic parameters, and has been observed via human subject studies in \cite{Yang:2015:PPC}.

In contrast to the earlier works, in this paper, using PT we address both user/network association and data pricing with emphasis on end-user behavior and decision making under uncertainty. We model the HetNet with probabilistic service guarantees and PT based decision making using the approach in \cite{Yousefvand:2018:IEU} where it was observed that the underweighting of the service guarantees by the users results in an increased rejection of the SP bids. To overcome this, a heuristic solution based on expanding the bandwidth available to the SPs bids was proposed in \cite{Yousefvand:2018:IEU}. However, it required the SPs to have knowledge of how the users perceive the uncertainty in the service guarantees, while in reality, SPs don’t have access to such information. The focus of the current paper is to overcome the bid rejection problem via a completely different approach, namely  a two-stage learning-based optimized bidding framework for SPs. In the first stage, we use a support vector machine learning algorithm to predict users’ binary decisions, and then in the second stage we cast the utility-optimized bidding problem as a Markov Decision Problem and propose a reinforcement learning-based dynamic programming algorithm to efficiently solve it. 

\section{System Model and Problem Formulation}
\label{sec:SystemModel}

\subsection{Network Model}
\label{NM}
To study user association in HetNets, we developed a two-tier HetNet scenario which includes $N$ wireless users that are randomly distributed within the coverage area of $K$ base stations. As shown in Fig. \ref{Fig:figl}, in our HetNet model there is one macrocell LTE BS located in the center of the area and $K-1$ overlaid small cell WiFi access points who are competing with each other to serve the users in the HetNet. We assume each user in the HetNet receives several bids from service providers (SPs) in both cellular and WiFi tiers, where each bid includes a data rate with a certain probabilistic service guarantee and at a given price. Upon receiving such bids, the user makes a binary decision to accept or reject each of the received bids.
\begin{figure}[tb!]
	\includegraphics[width=\linewidth]{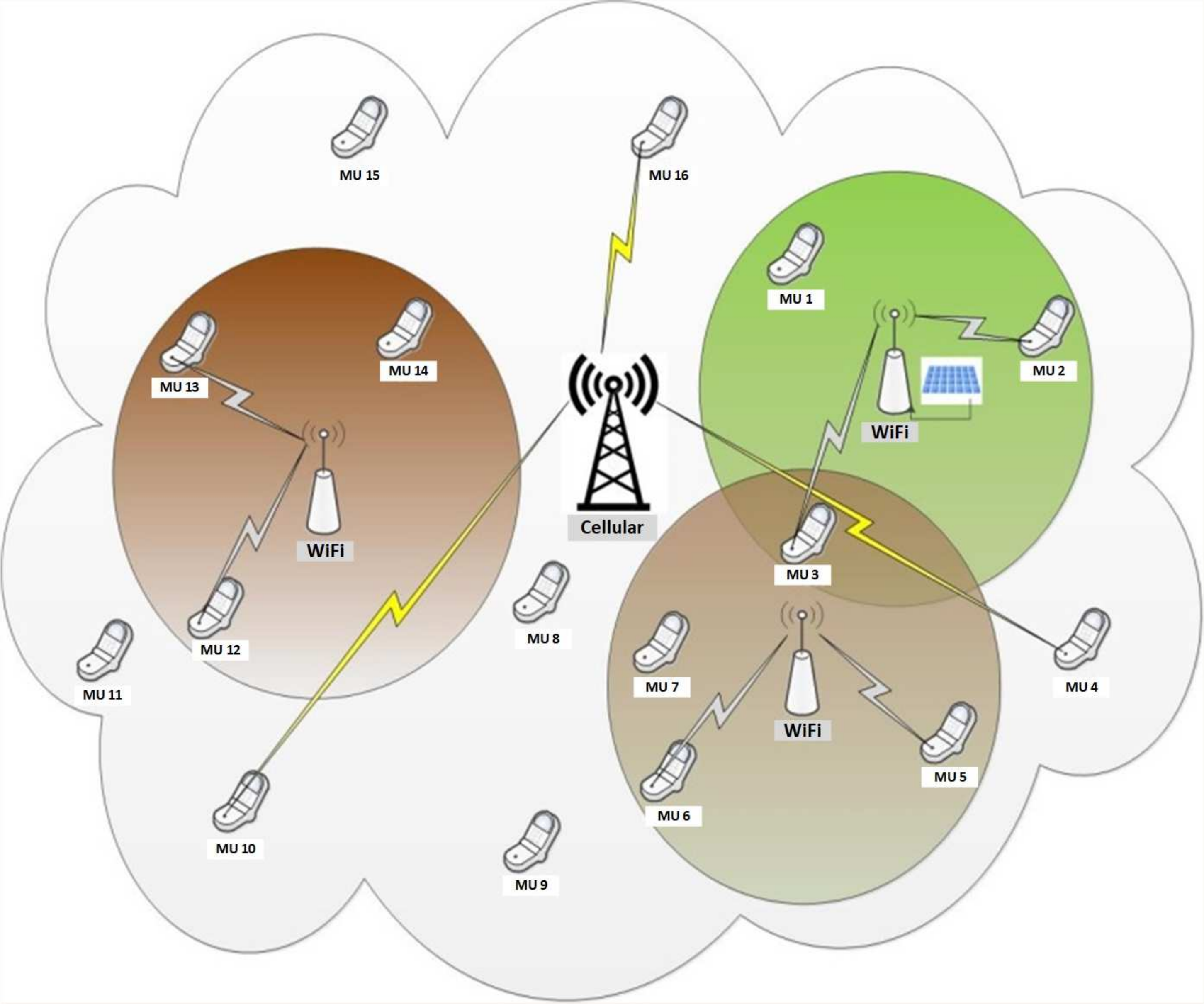}
	\centering
	\caption{Heterogeneous Network (HetNet) Model.}
	\label{Fig:figl}
\end{figure}

\subsection{Stackelberg Game for User Association in HetNets}
\label{subsec:SM}
In our HetNet model, each user receives $K$ different bids from all $K$ base stations, i.e., one bid from the cellular BS and $K-1$ bids from the WiFi BSs. To enable multihoming, we assume each user can be simultaneously connected to both cellular and WiFi SPs to receive the data service. Specifically, we assume that each user can only be associated to the cellular SP and the best serving WiFi SP for that specific user, which is the SP who offers a bid with highest utility among all WiFi SPs. Hence, we model user association problem in this work as a Stackelberg game with two leaders and one follower, where SPs act as the leaders who make service offers to the user, and the user serves as a follower who accepts or rejects the received bids. Since we have $N$ users in our HetNet model, to solve the user/network association problem in a distributed manner, we need to solve $N$ stackelberg games each with three players including WiFi SP and cellular SP as the leaders and one user as the follower. Assuming the BS's maximum bandwidth budget per user is fixed and the same for all users as shown in Eq. \ref{eq14}, these $N$ Stackelberg games are independent, and hence we consider and solve only one of these Stackelberg games without loss of generality. Note that finding the optimal bandwidth/power allocation for SPs is a NP-hard problem \cite{Yousefvand:2017:DES} and not the focus of this paper. In the remainder of this paper, we focus on one of these games.
Using the index $w$ for user's preferred WiFi SP and the index $c$ for cellular SP, we denote the $bids$ of WiFi and cellular SPs with triples $(b_w, r_{EUT}(b_w), BW_{w,EUT} )$, and $(b_c, r_{EUT}(b_c), BW_{c,EUT} )$, respectively, in which the first term shows the advertised data rate, the second term is the proposed price for the offered data rate, and the third term is the amount of BW that will be allocated to the user by each SP. User decisions are binary, which means the user either accepts a bid or rejects it, and there is no probabilistic decision by user. We denote user decisions on cellular and WiFi SPs bids with $p_c$ and $p_w$ respectively, which are binary variables. So, the tuple $(p_c,p_w)$ represents user’s strategy with regard to the received bids, hence, the user has four possible strategies $(0,0)$, $(0,1)$, $(1,0)$ and $(1,1)$. Fig. \ref{Fig:StackelbergGame} illustrates the Stackelberg game model between cellular and WiFi SPs and mobile user.
\begin{figure}[tb!]
	\includegraphics[width=\linewidth]{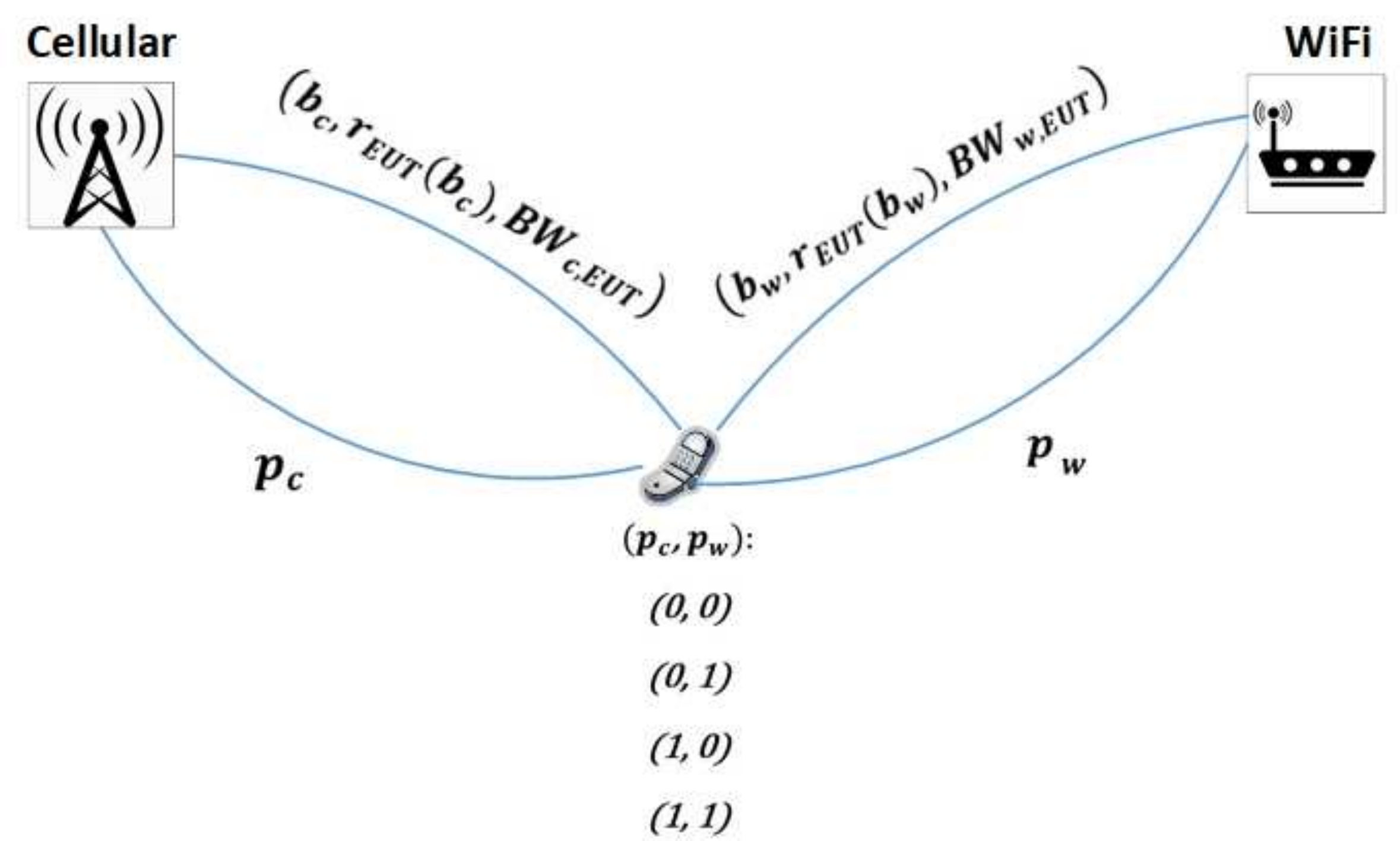}
	\centering
	\caption{Stackelberg Game Model.}
	\label{Fig:StackelbergGame}
\end{figure}

 All three players have a cost function and a benefit function, and their utility functions are simply the difference of the cost and benefit functions. We represent user’s utility function as
\begin{equation}
\label{eq1}
U_{user,{\rm{ }}EUT}{\rm{ }}\left\{ {p_c,p_w{\rm{ }}} \right\} = H\left( {B_{joint}{\rm{ }}} \right) - c_{user}{\rm{ }}\left( {p_w,p_c{\rm{ }}} \right),
\end{equation}
where $H(B_{joint})= \delta(B_{joint})^{1/\theta} $ is the users benefit function which is a logarithmic concave function of the user's expected aggregate data rate, $B_{joint}$, and $c_{user}{\rm{ }}\left( {p_w,p_c{\rm{ }}} \right)= p_w r_{EUT}(b_w) + p_c r_{EUT} (b_c)$ is the user’s cost function which shows the aggregate price that must be paid by user to the SPs for each $(p_c, p_w)$ strategy. The user's expected aggregate data rate is defined by 
\begin{equation}
B_{joint} = b_c \bar {F} _{B_c}( b_c,BW_{c,EUT})p_c + b_w \bar {F}_{B_w}( b_w,BW_{w,EUT})p_w, 
\end{equation}
 where $\bar F_{B_c}( b_c,BW_{c,EUT})$ is the service guarantee of the bid received from the cellular SP, and denotes the probability of having the actual data rate of user from cellular SP, $B_c$ which is a random variable, equal or higher than the advertised data rate by cellular SP, $b_c$, and is given by
\begin{equation}
\bar F_{B_c}( b_c,BW_{c,EUT})= Pr(B_c \geq b_c |BW_{c,EUT}),
\end{equation}
\noindent where for a fixed data rate $b_c$, a larger $BW_{c,EUT}$ yields a higher service guarantee.
Similar definition holds for $\bar F_{B_w}( b_w,BW_{w,EUT})$ which is the service guarantee of the bid received from WiFi SP.

Once the user chooses its best response strategy, $(p_c^*, p_w^*)$, the SPs will respond with their best response strategies to maximize their own utilities based on the received user decision. The utility of the WiFi SP, $U_{SP,w}$ is defined as
\begin{equation}
\label{eq2}
U_{SP,w} = p_w r_{EUT}( {b_w}) - C_w( b_w,BW_{w,EUT}),
\end{equation}
and the utility of cellular SP, $U_{SP,c}$ is defined as
\begin{equation}
\label{eq3}
U_{SP,c} = p_c r_{EUT}( {b_c}) - C_c( b_c,BW_{c,EUT}),
\end{equation} 
where, the first term in both of these equations is the SPs' expected payoff from the user, and the second term is their incurred service cost. The SPs' payoff from the user is equal to the offered price in their bids if the user accepts their bids, otherwise their payoff from the user is equal to zero. In this work, we assume both SPs use convex pricing functions, as $r_{EUT}(b_w)= {\alpha_1} (b_w)^{\beta_1}$, and $r_{EUT}(b_c)= {\alpha_2} (b_c)^{\beta_2}$, $({\beta_1},{\beta_2}>1)$, where ${\alpha_1}$ and ${\beta_1}$ are payoff parameters for the WiFi SP, and ${\alpha_2}$ and ${\beta_2}$ are payoff parameters for the cellular SP. We also assume the SPs use linear cost functions, as $C_w( b_w,BW_{w,EUT})= c_1 (b_w) + c_2(BW_{w,EUT})$, and $C_c( b_c,BW_{c,EUT})= c_3 (b_c) + c_4(BW_{c,EUT})$, where $c_1$ and $c_2$ are cost coefficients for the WiFi SP, and  $c_3$ and $c_4$ are cost coefficients for the cellular SP. To satisfy the user's minimum data rate constraint, the SPs must ensure that their offered data rate is higher than the minimum data rate required by the user, $b_{min}$. Thus, the data rate constraints for the WiFi and the cellular SPs will be defined as below, respectively:
\begin{align}
\label{eq4}
b_w \bar F _{B_w}( {b_w,BW_{w,EUT}}) \ge b_{min},\\
\label{eq5}
b_c \bar F _{B_c}( {b_c,BW_{c,EUT}}) \ge b_{min}.
\end{align}

\subsection{User Optimization Problem}
Upon receiving the bids from the SPs, the user will run an optimization problem to find its best strategy with regard to the received bids. We assume the user's payoff function from the received data is a concave function, as defined in Eq. \ref{eq1}, in which the user's utility is not linearly increased with increasing the data rate. It means that as long as the minimum data rate constraint is satisfied, the user is not willing to pay extra price with linear relation to the extra data rate offered by SPs. To find its best response strategy,$(p_c^*, p_w^*)$, user will run the following optimization problem (denoted as $Max1$): \\\\
\textbf {Max1 Problem:} \textit{User's Utility Maximization}.\\
---------------------------------------------------------------------------
\begin{align}
\label{eq6}
& \max_{p_c, p_w}~[\,\delta(B_{joint})^{1/\theta} -~p_w{\alpha_1} (b_w)^{\beta _1}-~p_c{\alpha_2} (b_c)^{\beta _2}~] && \\
& \text{subject to}\nonumber \\
\label{eq7}
& B_{joint}\ge~b_{min}  \\
\label{eq8}
&  \delta(B_{joint})^{1/\theta} \ge~~p_w{\alpha_1} (b_w)^{\beta _1}+~p_c{\alpha_2} (b_c)^{\beta _2}, \\
\label{eq9}
& p_c, p_w \in \{ 0,1\}
\end{align}

As shown above, the user has two major constraints for bid selection. The first constraint, shown in Eq. \ref{eq7}, is the user's data rate constraint which ensures the expected data rate for the user is higher than its minimum required data rate, $b_{min}$. The second constraint defined in Eq. \ref{eq8} is the user's utility constraint which guarantees a positive utility for the user from its strategy. 

\subsection{SPs Optimization Problems}
When the SPs receive user's decision with regard to their bids, they choose their best response strategy. The the best response strategy $(b_w^*, BW_{w,EUT}^*)$ for the WiFi BSs is obtained by solving the optimization problem below (denoted as $Max2$):\\\\
\textbf {Max2 Problem:} \textit{WiFi SP's Utility Maximization}.\\
---------------------------------------------------------------------------
\begin{align}
\label{eq10}
& \max_{b_w, BW_{w,EUT}}~[p_w {\alpha_1} (b_w)^{\beta _1} - (c_1 b_w + c_2 BW_{w,EUT})~] \\
& \text{subject to}\nonumber \\
\label{eq11}
& 0 \leq BW_{w,EUT} \leq BW_{w,max}, \\
\label{eq12}
& 0 \leq b_w \leq b_{w,max}, \\
\label{eq13}
&  b_w \bar F _{B_w}( {b_w,BW_{w,EUT}}) \ge b_{min},
\end{align}
in which, $BW_{w,max}$ is the maximum amount of bandwidth that can be allocated to the user by the WiFi SP, and $b_{w,max}$ is the maximum achievable data rate by the user from the WiFi SP considering $BW_{w,max}$ and the gain of the channel between the WiFi SP and the user. We assume the SPs use a proportionally fair bandwidth allocation algorithm to determine the maximum amount of bandwidth that can be allocated to each of their users. Assuming $BW$ as the total amount of bandwidth available at SP $i$, $i \in\{w,c\}$, and $N$ as the total number of users in our HetNet, the maximum amount of bandwidth that can be allocated to each user by the SP $i$, $BW_{i,max}$ is given as
\begin{equation}
\begin{aligned}
\label{eq14}
& BW_{i,max}= (G_{BA}*~BW)/\sum_{j=1}^{N}({a_j * c_j}),\\
\end{aligned}
\end{equation}
in which $G_{BA}$ is the gain of bandwidth allocation which is less than one due to the guard bands between channels for preventing interference, $a_j$ is a binary variable showing the activity of the users which $a_j=1$ if the user $j$ is active and has data demand, otherwise $a_j=0$,  and $c_j$ is also a binary variable representing the coverage of the user $j$ by the the BS $i$, where $c_j=1$ if the SINR of the link between the user $j$ and the BS $i$ is higher than a certain threshold to have a reliable transmission, otherwise $c_j=0$. Considering $BW_{i,max}$, the maximum achievable data rate by user, $b_{i,max}$ is given by:

\begin{equation}
\begin{aligned}
\label{eq15}
& b_{i,max}= BW_{i,max} \log (1+\frac{P h_j a_j c_j}{\sigma^2}),\\
\end{aligned}
\end{equation}
where, $P$ is the transmit power of the BS, $h_j$ is the channel gain between the user $j$ and the SP $i$, and $\sigma^2$ is the noise variance over the transmission channel. Similarly, the cellular SP runs $Max3$ optimization problem, which is defined exactly similar to $Max2$, except the index $w$ is replaced with the index $c$ and the parameters $\alpha_1$, $\beta _1$, $c_1$, $c_2$ are replaced with the parameters $\alpha_2$, $\beta _2$, $c_3$, $c_4$, respectively. In Theorem \ref{Theorem1}, we prove that for both WiFi and cellular SPs, the best response strategies derived from $Max2$ and $Max3$ optimization problems, satisfy the minimum data rate constraint with equality. 
\begin{theorem}
	\label{Theorem1}
	The SPs best response strategies, derived from $Max2$ and $Max3$ problems, will always satisfy the minimum data rate constraint in the boundary of its feasibility region, i.e. we always have $b_w^* \bar F _{B_w}( {b_w^*,BW_{w,EUT}^*}) = b_{min}$, and $b_c^* \bar F _{B_c}( {b_c^*,BW_{c,EUT}^*}) = b_{min}$ for the WiFi and the Cellular SPs, respectively.
\end{theorem}

\begin{proof}
	We prove this by contradiction for the WiFi SP.
	Assume $({b_w^*,BW_{w,EUT}^*})$ is the optimal solution for $Max2$ problem, and $BW_{w,EUT}^*$ is not a marginal BW, i.e. $b_w^* \bar F _{B_w}( {b_w^*,BW_{w,EUT}^*}) \neq b_{min}$. Considering Eq. \ref{eq13}, we can infer that $b_w^* \bar F _{B_w}( {b_w^*,BW_{w,EUT}^*}) > b_{min}$ (1).
	In this case $\exists~BW_{w,EUT}^\prime $ such that $b_w^* \bar F _{B_w}( {b_w^*,BW_{w,EUT}^\prime}) = b_{min}$ (2). From (1) and (2), and considering the direct relation between $F _{B_w}( {b_w^*,BW_{w,EUT}^*})$ and $BW_{w,EUT}^*$, we can infer that $BW_{w,EUT}^\prime<BW_{w,EUT}^*$. Now, considering the WiFi SP's utility function defined in Eq. (11) and due to its reverse relation with the advertised bandwidth, $BW_{w,EUT}$, we can infer that $U_{SP,w}({b_w^*}, BW_{w,EUT}^\prime)>U_{SP,w}({b_w^*}, BW_{w,EUT}^*)$ which is in conflict with the initial assumption that $({b_w^*,BW_{w,EUT}^*})$ is the optimal solution for $Max2$ problem. So, the proof is complete. The same proof is valid for the cellular SP.
\end{proof}
A characterization of the Nash Equilibria under the EUT
and the PT models follows directly from the above Theorem,
and can be found in \cite{Yousefvand:2018:IEU}. 
To model the effects of PT on user decision making, we assume that users use their subjective biases in evaluating the bids made by the SPs. Specifically, we use the probability weighting effect (PWE) of PT to model how the user evaluates the probabilistic service guarantee that is part of the SPs bids. We use the Prelec function \cite{Prelec:1998:PWF} to model the PWE under PT, which is given by
\begin{align}
\label{perlec}
&& w(p)=exp(-{(-ln(p)}^\alpha ),~~(0 < \alpha< 1 ),
\end{align}
\noindent where $w(p)$ is the weighted version of the probability $p$, which is used by end-users while making decisions based on probabilistic parameters, and $\alpha$ is the Prelec function parameter that describes the deviation of $w(p)$ from $p$. This function is a regressive and s-shaped function which is concave in $0<p<1/e$ region and convex in $1/e<p<1$ region, and $w(p)>p$ in the former domain while $w(p)<p$ in the later. Under this PWE model, we can infer that the user overestimates the service guarantees of the received offers if the advertised service guarantees are less than $1/e=0.37\%$, and user underestimates them if advertised service guarantees are higher than $0.37\%$.  Assuming that SP networks are well designed to
offer service guarantees higher than $1/e$, we focus on the underestimating of service guarantees by the user under PT. It is also justified by the fact that end-users in real world wireless networks typically perceive the quality of their service as lower than that advertised by the SPs \cite{Measuring:online:MMB,Measuring:online:MBA}.

The results in \cite{Yousefvand:2018:IEU} reveal that when users underweight the SPs advertised service guarantees, the rejection rate of the SP bids, and consequently the resulting utility and revenue for SPs decreases dramatically. To overcome this, we discuss next a learning-based optimized bidding mechanism for SPs.

\section{Learning-based Optimized Bidding Method for SPs}
\label{sec:LBOB}
 To find the utility-optimized bids, we propose a two-stage learning method for SPs where in the first stage, SPs learn a binary user decision classifier function via a support vector machine learning algorithm. Then in the second stage, SPs find the utility-optimized bids using a reinforcement learning-based bidding algorithm in which the classifier function obtained in the first stage is used to evaluate the value of each potential bid by predicting the user decision with regard to that bid. A heuristic solution based on expanding the offered bandwidth in the SPs bids was proposed in \cite{Yousefvand:2018:IEU}. However, it requires the SPs to have knowledge of how the users perceive the uncertainty in the service guarantees, while in reality, SPs don’t have access to such information.
\subsection{SVM-based Classifier for Predicting User Decisions}
\label{sec:SVM}
Using a record of previous user decisions in response to different offered services, and noting the binary nature of user decisions in our work, we can efficiently train a SVM classifier to predict user decisions with regard to any future bid. The set of training data, here denoted as D, is given by
\begin{equation}
D=\{(x_1, y_1 ), (x_2, y_2 ), …  (x_N, y_N)\}      
\label{eq16}  
\end{equation}

\noindent which includes $N$ entries of labeled data samples $(x_i, y_i),$ where the feature vector $x_i=(b_i, r(b_i), BW_i)$ represents the bid $i$ features including its advertised data rate, the offered price, and offered bandwidth, respectively, and $y_i$ represents the previous binary decision of user in response to bid $i$. Using the gathered training data, the SVM classifier parameters, $w$ and  $b$, can be derived by solving the optimization problem given in (\ref{eq18})-(\ref{eq20}) using the SVM learning algorithm. \\\\
\textbf {SVM Learning Optimization}.\\
---------------------------------------------------------------------------
\begin{align}
\label{eq18}
& \min_{w, b, \epsilon}~1/2~ w^{T}w + c \sum_{i=1}^{N} \epsilon_i, &\\
& \text{subject to:}\nonumber& \\
\label{eq19}
& y_i~(w^T \phi {(x_i )}+b) \geq 1-\epsilon_i, &\\
\label{eq20}
& \epsilon_i \geq 0. &\\
\nonumber
\end{align}
In the SVM optimization objective function given in Eq. (\ref{eq18}), $2/w^T w$ is the width of the separating margin between two groups of accepted and rejected bids which is supposed to be maximized by minimizing its inverse, and $\epsilon_i$ is the amount of error in classifying each data point $i$, which is supposed to be minimized, and $c$ is a regularization parameter that defines the trade-off between these two objectives. The inequality constraint in Eq. (\ref{eq19}) enforces $w^T \phi {(x_i )}+b\geq 1-\epsilon_i$ for all accepted bids and $w^T \phi {(x_i )}+b\leq -(1-\epsilon_i)$ for all rejected bids, as it combines these two constraints.
The constraint in Eq. (\ref{eq20}) ensures that the amount of classification error for each data point $i$, $\epsilon_i$, is either zero for correctly classified points, or a positive value for misclassified points. After finding the SVM classifier parameters $w$ and $b$ from solving the SVM optimization, we can define the SVM classifier function $f(x_i)$ as

\begin{equation}
f(x_i )=w^T x_i+b= \begin{cases}
 \geq 0, & \text{i.e.}\ d(x_i)=1,~ accepted, \\
\leq 0, & \text{i.e.}\ d(x_i)=0,~rejected,
\end{cases}
\label{eq21}
\end{equation}
to predict the user's binary decision with regard to any given bid $x_i$, denoted as $d(x_i)$.

\subsection{Reinforcement Learning-based Optimized Bidding}
\label{sec:LOB}
To find the utility-optimized bids for SPs that are highly likely to be accepted by the user while maximizing their utility, we formulate the SPs optimized bidding problem as a Markov Decision Problem (MDP) with the tuple $(S,A_S,P_{SA},\gamma,R )$ in which $S$ denotes the set of states, $A_S$ defines the set of actions in each state, $P_{SA}$ is the set of state-action transition probabilities, $\gamma$ is the discount factor, and $R$ is the reward function. In this work, the set of states is defined as
\begin{equation}
S=\{s_1,s_2,…,s_{M*N} | s_i=(b_{s_i }, BW_{s_i } )\}
\label{eq22}
\end{equation}
\noindent in which, each state $s_i$ is associated with a tuple $(b_{s_i }, BW_{s_i })$ which is a potential bid and shows how much data rate, and bandwidth will be advertised to the user, respectively, if SPs bid $s_i$. Note that to cast the SPs bidding problem as a MDP, we need to have discrete set of states, and to do so we quantize the possible values for both originally continuous bid variables bandwidth, $BW$, and data rate, $b$, by considering $M$ possible values for data rate, and $N$ possible values for bandwidth. Since there is one state defined for each possible combination of the values of $b$ and $BW$, we have overall $M*N$ states in our MDP model, considering all such possible combinations.

\noindent $A_{S} = \{ A_{s_i} \}_{i=1}^{MN}$ is the set of  is the set of actions available in each state $s_i\in S,i=1,2,…,M*N$ . Taking each action $a_j$ in a given state $s_i$ results in a transition from state $s_i$ to the new state $s_j$ , with new data rate and bandwidth coordinates of $(b_{s_j}, BW_{s_j} )$. Hence, each action can be defined in terms of how that action changes the current data rate, and bandwidth coordinates of the system, and hence the action set for any given state  $s_i$, $A_{s_i}$, is given by
\begin{align}
\label{eq23}
&A_{s_i}=\{a_1,a_2,…,a_{M*N} | a_j=(\Delta b_j, \Delta BW_j)& \nonumber\\
&=(b_{s_j}- b_{s_i}, BW_{s_j} - BW_{s_i}) \}&
\end{align}
\noindent in which, there is one unique action defined for each possible transition from any state $s_i$ to any other state $s_j$. For example, taking action $a_j$ in state $s_i$ changes the state of the MDP from the state $s_i$ to the state $s_j$ by adding the $(\Delta b_j, \Delta BW_j)$ to the coordinates of the state $s_i$.  Equation (\ref{eq24}) explains how this transition, denoted as $s_i \overset{a_j}{\rightarrow} s_j$, happens between these two states in the MDP.
\begin{align}
\label{eq24}
&s_i \overset{a_j}{\rightarrow}: (b_{s_i},BW_{s_i})+(\Delta b_j, \Delta BW_j))\nonumber \\ &=(b_{s_i},BW_{s_i})+(b_{s_j}- b_{s_i}, BW_{s_j} - BW_{s_i})=(b_{s_i},BW_{s_i})=s_j 
\end{align}

\noindent$P_{SA}$ defines the transition probabilities for all state-action pairs, which in our model are deterministic and given as 
\begin{equation}
P_{SA}=\{Pr(s' | s_i,a_j ), \forall s_i,s' \in S , \forall a_j \in A_{s_i})\}
\label{eq25}
\end{equation} 

\noindent where 
\begin{equation}
Pr(s' | s_i,a_j )=\begin{cases}
1, &\text{if}\ s'= s_j, \\
0, &\text{if}\ s'\neq s_j.
\end{cases}
\label{eq26}
\end{equation}

\noindent $R(s_i,a_j )$ is the state-action reward function which defines an immediate reward that can be achieved from taking the action $a_j$ in state $s_i$. In our model, $R(s_i,a_j )$ is defined as
\begin{equation}
R(s_i,a_j )= d(s_j)* \{r(b_{s_i} ) )-c(b_{s_j}, BW{s_j} )\},   
\label{eq27}
\end{equation}
where $d(s_j)$ is the users predicted decision with regard to bid $s_j$ which is defined in Eq. (\ref{eq21}), $r(b_{s_i} )$ and $c(b_{s_j}, BW{s_j} )$ are the advertised price and the expected cost associated with bid $s_j$, respectively, which are defined in section \ref{subsec:SM}. In our model, the reward of taking any action $a_j$ in a given state $s_i$ to go to the new state $s_j$ is equal to the expected utility of bidding the bid $(b_{s_j}, BW{s_j})$ which is associated to the destination state $s_j$. \\
The goal for SPs is to find a bid which maximizes their utility, which in our MDP model is equivalent to finding the best state with an associated bid with the highest amount of reward or expected utility. Such a state is accessible from any initial state in the MDP, by taking a set of actions based on the optimal policy. A policy is any function $\pi:S\rightarrow A$ mapping from the states to the actions, and we say that we are executing some policy $\pi$ if whenever we are in any state $s$, we take action $a=\pi (s)$. The value function $V_\pi (s_i)$ for evaluating the value of each policy $\pi$ in each state $s_i$ is defined as

\begin{equation}
V_\pi (s_i )=R(s_i,\pi(s_i ))+\gamma \sum_{s'\in S} pr(s'| s_i,\pi(s_i ) )~V_\pi (s' )  
\label{eq28} 
\end{equation}
\noindent assuming $\pi(s_i )=a_j$, and since transition probabilities are binary and given in our MDP model, the value function can be simplified to

\begin{equation}
V_\pi (s_i )=R(s_i,a_j)+\gamma V_\pi (s_j )
\label{eq29}
\end{equation}
as in our model we have $Pr(s'|s_i,a_j )=0~\forall s'\neq s_j,$ and $Pr(s'|s_i,a_j )=1~if~s'=s_j$. Also, since the optimal state is accessible from any initial state in the MDP, we set the discount factor $\gamma =0$ to focus on the current action only and find the optimal action that takes us from the current state to the optimal state. Such an action is given by executing the optimal policy, $\pi^*$, in any given state $s_i$, as
\begin{equation}
\pi^* (s_i)= \underset{a_j \in A_{s_i}}{\operatorname{argmax}}~ [R(s_i, a_j)].
\label{eq30}
\end{equation}

To find the optimal policy which defines the best action for each state, we should theoretically exhaust all possible actions in each state, which is computationally intractable for MDPs with large number of states and actions. To solve this issue, we propose an efficient dynamic programming based algorithm which finds the optimal action for each state in the formulated MDP, with a logarithmic convergence rate.

\subsection{Dynamic Programming-based Optimized Bidding (DPOB) Algorithm}
\label{sec:LOB}
In the formulated MDP for optimized bidding, if user rejects any bid $s_i$, it will also reject any other bid $s_j$ that results in a lower expected utility for that user, considering the rationality of users and their objective to maximize their expected utility. In fact, if user rejects the bid $s_i=(b_{s_i}, BW_{s_i})$ it will reject any other bid $s_j=(b_{s_j}, BW_{s_j})$ if $(b_{s_i} \leq b_{s_j}) \land (BW_{s_i} \geq BW_{s_j})$ since $s_i$ Pareto dominates $s_j$ from the user's perspective, under such conditions. Likewise, when user accepts any bid $s_i$, the offering SP is not willing to offer any new bid $s_j$ if $s_j$ results in a lower utility for that SP. In fact, if user accepts any bid $s_i=(b_{s_i}, BW_{s_i})$, the offering SP will not offer any new bid $s_j=(b_{s_j}, BW_{s_j})$ to the user if $(b_{s_i} \geq b_{s_j}) \land (BW_{s_i} \leq BW_{s_j})$ since $s_i$ Pareto dominates $s_j$ from the SP's perspective, under such conditions. Note that based on the utility functions of user and SPs defined in section \ref{subsec:SM}, when other parameters remain the same, increasing the advertised data rate will increase the utility of SPs and decrease the utility of users, while increasing the allocated bandwidth increases the utility of users and decreases the utility of SPs. The Pareto optimal relation between different states in the MDP and their associated bids helps SPs to update the set of feasible actions in each state after taking any random action, by removing the infeasible actions from the list of actions in the newly visited state while executing the reinforcement learning algorithm. The proposed dynamic programming-based reinforcement learning algorithm for optimized bidding is given in Algorithm \ref{algorithm1}. Note that in Algorithm \ref{algorithm1}, instead of initializing the MDP with a random state/bid, we initialize it with the SP's previous bid derived from MAX2 optimization under EUT model, here denoted as $\bar{s}=(b^*_{EUT}, BW^*_{EUT})$, to increase the efficiency of the algorithm. The iterative pruning of infeasible actions from the list of actions in the proposed dynamic programming-based optimized bidding algorithm, leads it to converge to the optimal solution, and find the utility-optimized bids with a logarithmic convergence rate according to Theorem \ref{Theorem2}.\\\\
\textbf{Algorithm1:} \textit{Dynamic Programming-based Optimized Bidding (DPOB) Algorithm}\\
--------------------------------------------------------------------
\begin{algorithm}
\KwIn{\{Set of states: $S$,~Initial state of the system: $\bar{s}=(b^*_{EUT}, BW^*_{EUT})$,~Initial set of actions:$A_{\bar{s}}$\}}

\textbf{Initialize} $\textit{SP\_Utility=0}$, $s^*=(0,0)$ \\
\While{$A_{\bar{s}} \neq \emptyset$}
{
	choose a random action $a_j$ from the set $A_{\bar{s}}$\\
	calculate $R(\bar{s}, a_j)$ using Eq. (\ref{eq27})\\
	\If{$R(\bar{s}, a_j) \geq \textit{SP\_Utility}$}
	{
		\textbf{Set} $\textit{SP\_Utility} = R(\bar{s}, a_j)$\\
		\textbf{Set} $s^*=(b_{s_j},BW_{s_j})$ 
	}
	\textbf{Update set of actions:}\\
	{
		Calculate $d(s_j)$ using SVM classifier in Eq. (\ref{eq21})\\ 
		\If {$d(s_j)=1$} {\textbf{Update} S: remove all actions $a_k$ from $A_{\bar{s}}$ 
			if $(b_{s_k} \leq b_{s_j}) \land (BW_{s_k} \geq BW_{s_j})$}
		\If {$d(s_j)=0$} {\textbf{Update} S: remove all actions $a_k$ from $A_{\bar{s}}$ 
			if $(b_{s_k} \geq b_{s_j}) \land (BW_{s_k} \leq BW_{s_j})$}
	}
}

\KwOut{$s^*$, $\textit{SP\_Utility}.$}
\label{algorithm1}
\end{algorithm}

\begin{theorem}
\label{Theorem2}
The proposed DPOB algorithm converges to the optimal solution with a logarithmic convergence rate of $\mathcal{O} (\log_{1.33} |S| )$.
\end{theorem}

\begin{proof}
	We first show that DPOB algorithm converges to the optimal solution. We prove this by contradiction. Assume $s_{j^*} \in S$ is the optimal state in the MDP such that its associated bid $(b_{s_{j^*}}, BW_{s_{j^*}})$ maximizes the SPs utility, i.e. $a_{j^*}= \underset{a_j \in A_{\bar{s}}}{\operatorname{argmax}}~ [R(\bar{s}, a_j)]$, assuming $\bar{s}$ to be the initial state of the MDP. Since the DPOB algorithm exhausts all possible transition actions from initial state to all other states, the optimal state $s_{j^*}$ will be visited by the algorithm unless the action $a_{j^*}$ gets removed from the set of actions, $A_{\bar{s}}$, during iterative updates of this set in the algorithm. In that case, there must exists a state $s_k \in S $ which Pareto dominates $s_{j^*}$ and results in higher utility for the SP as compared to $s_{j^*}$, which is in contrast with initial assumption that $s_{j^*}$ is the optimal state in the MDP.
	Now we have to show that DPOB algorithm converges to the optimal solution with a logarithmic convergence rate of $\mathcal{O} (\log_{1.33} |S| )$. Assume $s_l \in S$ and $s_h \in S$ to be the states with lowest and highest expected utility for the SP, respectively. If SP bids $s_l$ and user accepts that, still all the other actions are feasible, i.e. no action will be removed from $A_{\bar{s}}$ after its update, while if SP bids $s_h$ and user accepts that, all other actions become infeasible as $s_h$ Pareto dominates all other states, which means the list of actions $A_{\bar{s}}$ becomes empty immediately. However, in general after visiting any new state like $s_j$, on average for half of the remaining feasible states like $s_k$ we either have $(b_{s_k} \leq b_{s_j}) \land (BW_{s_k} \geq BW_{s_j})$, or $(b_{s_k} \geq b_{s_j}) \land (BW_{s_k} \leq BW_{s_j})$, which means either the state $s_j$ Pareto dominates them, or they will Pareto dominate the state $s_j$, and depending on the users response, one group of such states become infeasible, and their associated transition actions will be removed from the set of actions, $A_{\bar{s}}$. Hence, on average the initial set of actions $A_{\bar{s}}$ shrinks by one fourth after each iteration of the algorithm, and become empty after $\mathcal{O} (\log_{(4/3)} |S|) $ iterations which proves the converges to the optimal solution with $\mathcal{O} (\log_{1.33} |S| )$, where $|S|$ denotes the total number of states in the MDP.
\end{proof}

  According to Theorem \ref{Theorem2}, the complexity of the proposed DPOB algorithm is $\mathcal{O} (\log_{1.33} |S| )$ while model-free reinforcement learning algorithms like Q-learning have a complexity of $\mathcal{O} (ne )$ where $n$ denotes the number of states, and $e$ the number of actions \cite{Koenig:1992:CAR}. For the bidding problem considered here, we have $e=n=|S|,$ and it follows that the complexity of a Q-learning algorithm used to solve it would be $\mathcal{O} (|s|^2)$. Thus, the computational gain of DPOB as compared to model free reinforcement learning algorithms is noteworthy since it reduces the network access delay for users which is attractive for SPs to enable and support delay critical applications in 5G HetNets.

\section{Simulation Results}
\label{sec:SimulationResults}
In this section, we provide simulation results to validate the efficiency of the proposed learning based bidding algorithm. As shown in Fig. \ref{fig2}, we consider a HetNet scenario in which there are $U$ randomly distributed mobile users, that are covered by 9 BSs including one cellular BS (CBS) with the coverage radius of $1000 ft$ located in the center of a macro-cell and 8 overlaid small-cell long-range WiFi BSs (WBSs) with maximum coverage radius of of $300 ft.$ Each WBS competes with the CBS to offer data service to each mobile user located in a coverage area common to them. We use the Hata propagation model for urban environments \cite{Rappaport:2002:Hata} to capture the effects of path loss on each of the radio links. Table \ref{table1} details the parameters used in the simulation of the HetNet environment modeled here.
\begin{figure}[htb!]
	\includegraphics[width=\linewidth]{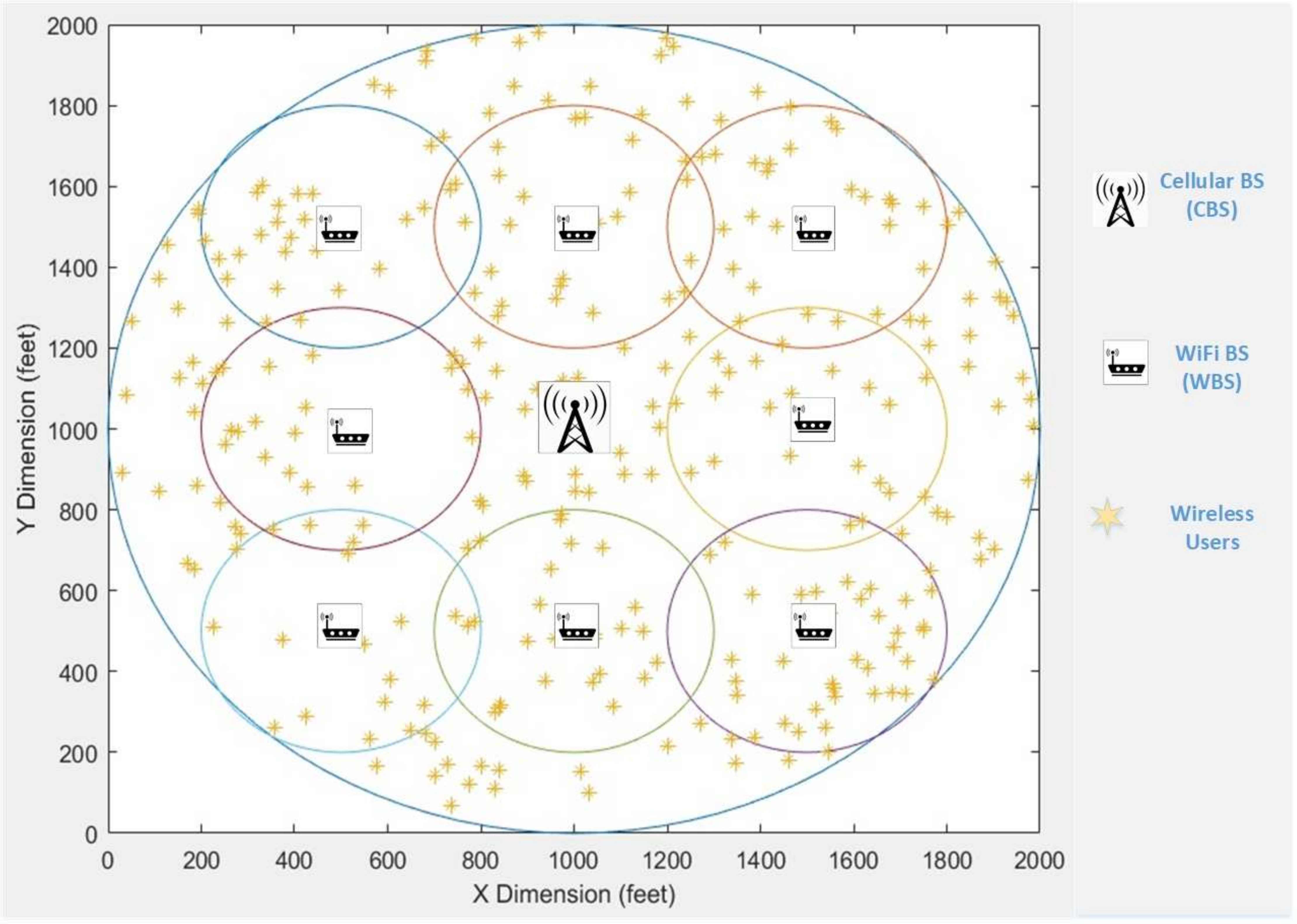}
	\centering
	\caption{Simulated HetNet Scenario.}
	\label{fig2}
\end{figure}

\begin{table}
	\includegraphics[width=\linewidth]{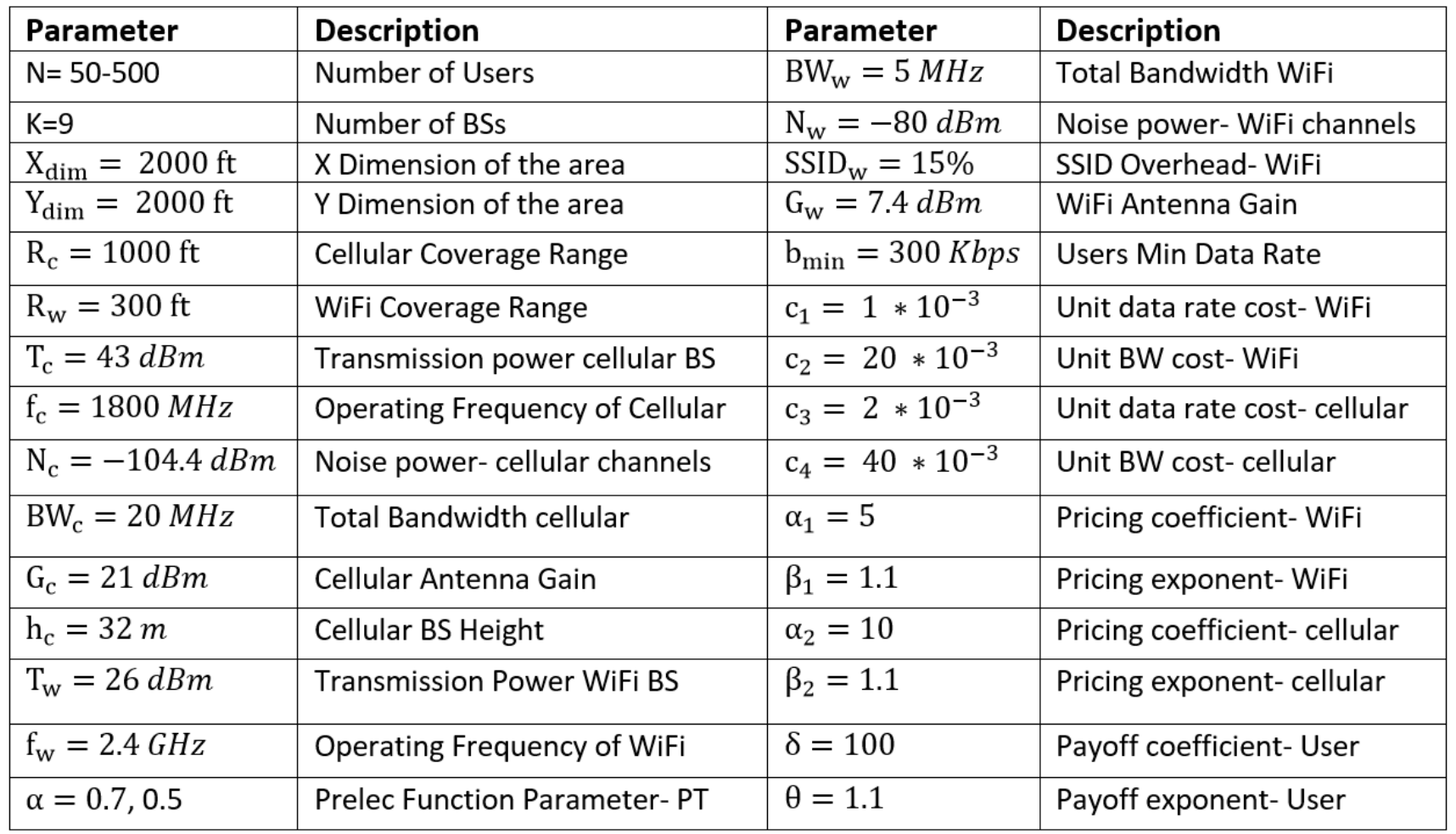}
\centering
\caption{Simulation Parameters.}
\label{table1}
\end{table}
Fig. \ref{fig3} compares the sum utility of the SPs under three scenarios: (i) EUT- where there is no difference between the users objective and subjective perceptions of the SPs service guarantee, (ii) Deviation from EUT- where the users subjectively weight the SPs service guarantees with the Prelec probability weighting function parameterized by $\alpha$, and (iii) DPOB- where the learning algorithm is used to predict the users responses to bids. In the figure, in each of the cases,  the total number of users (load) in the HetNet changes from $U = 50$ to $U = 500$. As we can see, when the number of users is lower than 200, since the advertised data rates are high and low service guarantees in the range of $[0, 1/e]$ are advertised by SPs to satisfy users data rate constraints, users overestimate the value of the SPs bids due to PWE of PT described in Eq. (18). When the total number of users is more than 200, since the advertised data rates are low, the SPs bids include high service guarantees in the range of $[1/e, 1]$ which results in the underestimation of such bids by the users and their subsequent rejection. Thus the deviation from EUT affects the network adversely in the region corresponding to underweighting of the service guarantees. It can also be seen that the proposed DPOB algorithm uniformly improves performance in both low and high load situations, by taking (a) advantage of learning the end users decision making, and (b) optimizing the SPs bids based on the learned information to maximize the users acceptance rate of the offered bids. In fact, as seen in Fig. \ref{fig3}, the sum utility of SPs under DPOB increases on average by a factor of $3.27$ as compared to EUT model.

Further, the gap between DPOB and EUT is bigger in low load situations. The reason is that when the total number of users is low, the SPs allocate more bandwidth to each user and hence their advertised data rates are much higher, which in turn results in a very high price for their bids considering the convex pricing function used by them. This situation results in a rejection of bids by the users since their payoff function is a concave function of the data rate, making the SPs offers unattractive. In the case of DPOB, the ability to learn user actions (accept/reject) by evaluating all feasible bids, allows the SPs to make only those bids that are still affordable to the users. 

Fig. \ref{fig4} compares the sum utility of users under the mentioned three scenarios. As we can see again, the proposed DPOB algorithm uniformly outperforms EUT model in terms of users sum utility, due to offering more affordable bids to users which increases user's acceptance rates and utilities, accordingly. Moreover, we observe that users sum utility under all three models increases with increasing the number of users, before reaching the max network capacity under our setting when $U=400$. The reason is that when the load is low, users receive bids with advertised data rates much higher than their required minimum data rate, and due to SPs convex pricing and users concave payoff functions, such bids result in low utility for users, while by increasing the load, the gap between advertised data rates of SPs and required data rates of users shrinks, which increases the users utility accordingly. As seen in Fig. \ref{fig4}, the sum utility of users under DPOB increases on average by a factor of $1.65$ as compared to the EUT model.

We can also see in both Fig. \ref{fig3} and Fig. \ref{fig4} that when the number of users is more than 400, which is the max network capacity under our setting, most of the SPs are not capable of satisfying the users data rate constraints since their BW budget per user decreases with increasing load. Therefore, the number of users who accept bids, and are connected to BSs decreases, which in turn reduces both the SPs and users sum utilities.
\begin{figure}[htb!]
	\includegraphics[width=\linewidth]{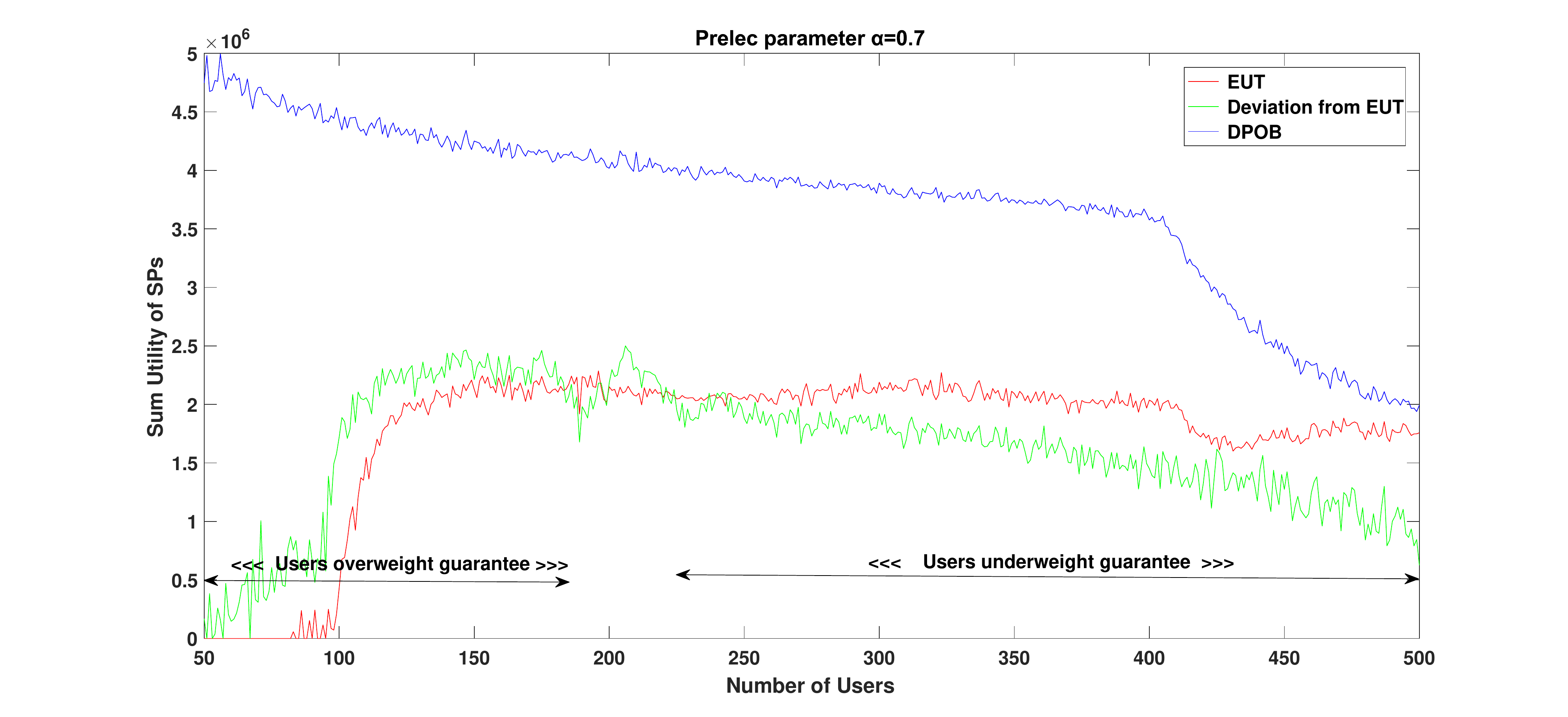}
	\centering
	\caption{Sum Utility of SPs under EUT, Deviation from EUT and DPOB.}
	\label{fig3}
\end{figure}

\begin{figure}[htb!]
	\includegraphics[width=\linewidth]{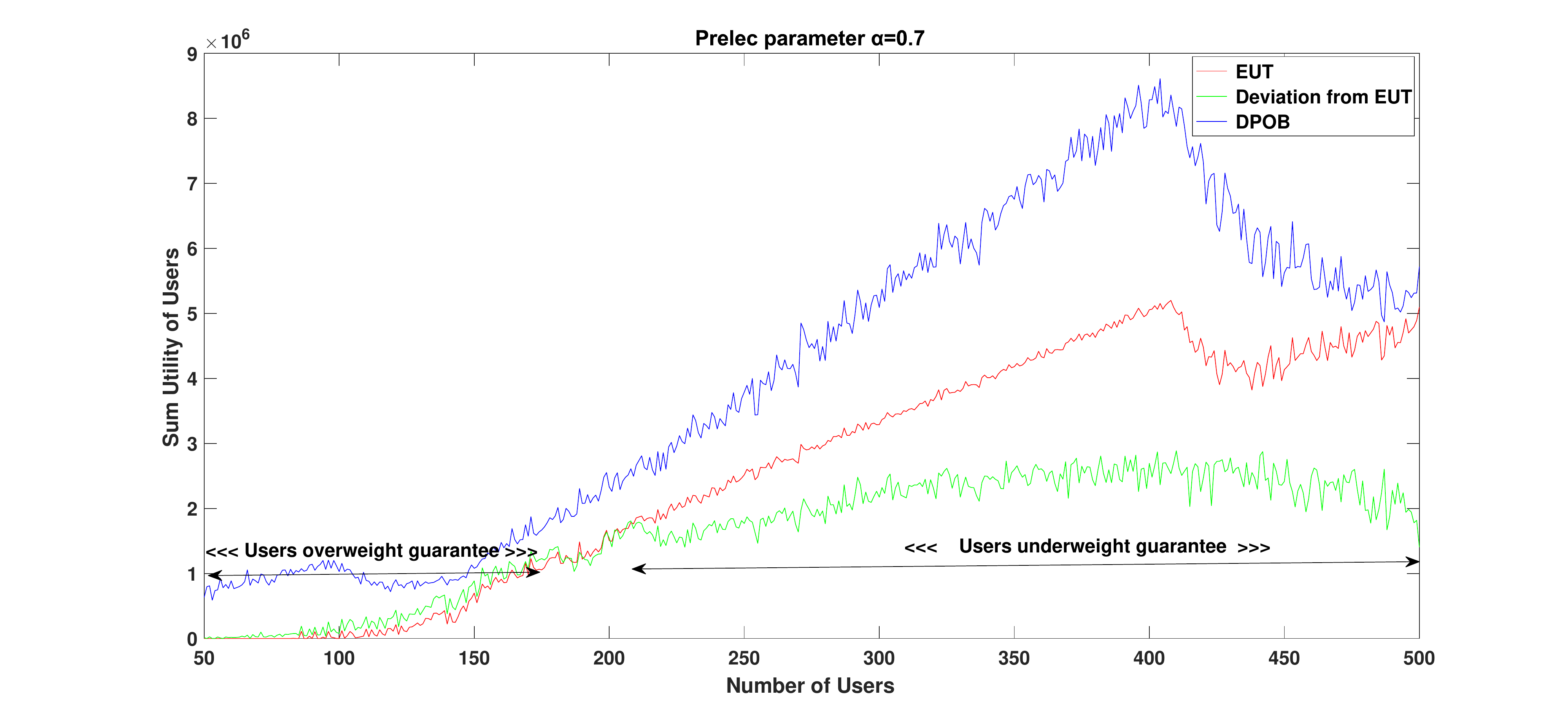}
	\centering
	\caption{Sum Utility of Users under EUT, Deviation from EUT and DPOB.}
	\label{fig4}
\end{figure}

Fig. \ref{fig5} and Fig. \ref{fig6} compare the sum utility of SPs and sum utility of users for the three scenarios, respectively, when the Prelec parameter $\alpha= 0.5$. This corresponds to a situation where the PWE of PT is more intense, i.e. the deviation of end-user behavior from EUT is greater. The performance results in the figures show that the deviation from EUT is more pronounced and the performance improvements in sum utilities using DPOB are also more significant. 

\begin{figure}[htb!]
	\includegraphics[width=\linewidth]{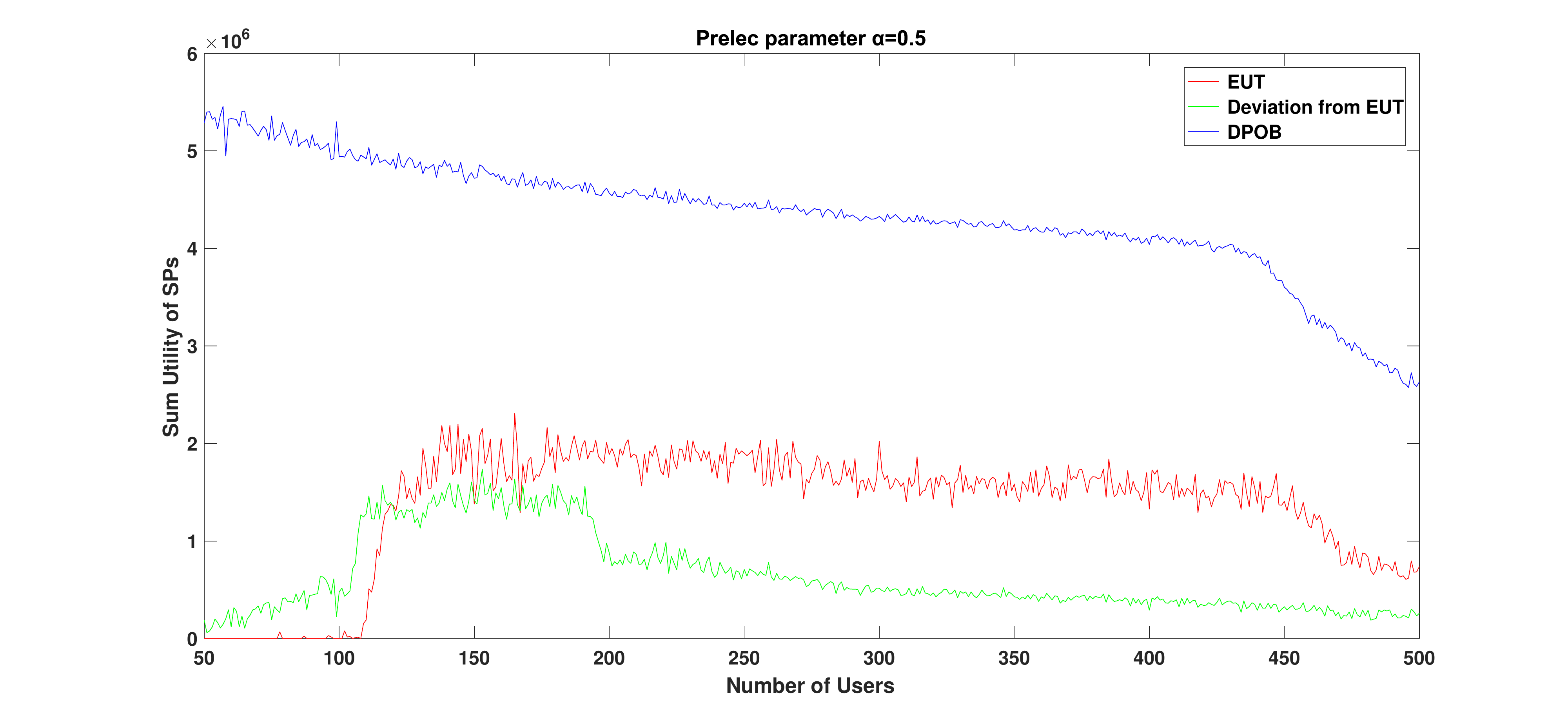}
	\centering
	\caption{Sum Utility of SPs under EUT, Deviation from EUT and DPOB.}
	\label{fig5}
\end{figure}

\begin{figure}[htb!]
	\includegraphics[width=\linewidth]{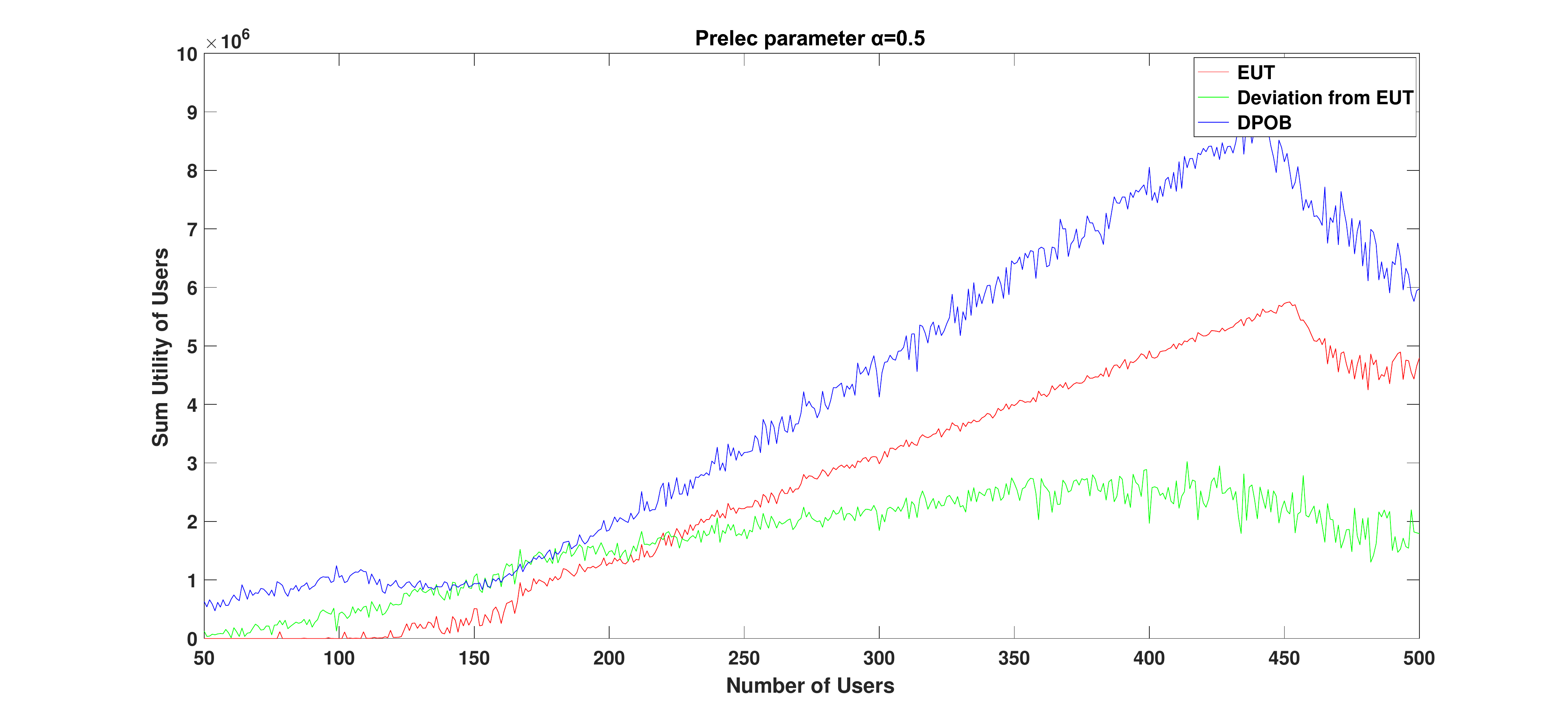}
	\centering
	\caption{Sum Utility of Users under EUT, Deviation from EUT and DPOB.}
	\label{fig6}
\end{figure}

\section{Conclusion}
\label{sec:Conclusion}
In this paper, we studied the impact of end-user behavior on SP bidding and user/network association in a HetNet with multiple SPs while considering the uncertainty in the service guarantees offered by the SPs. We formulated user association with SPs as a multiple leader Stackelberg game where each SP offers a bid to each user that includes a data rate with a certain probabilistic service guarantee and at a given price, while the user chooses the best offer among multiple such bids. Using PT to model end-user decision making that deviates from EUT, we showed that when users underweight the advertised service guarantees of the SPs, the rejection rate of the bids increases dramatically which in turn decreases the SPs utilities and service rates. To overcome this, we proposed a two-stage learning-based optimized bidding framework for SPs. In the first stage, we used a support vector machine learning algorithm to predict users' binary decisions, and then in the second stage we cast the utility-optimized bidding problem as a MDP and proposed the DPOB algorithm to efficiently solve it. Simulation results and computational complexity analysis validated the efficiency of the proposed learning based bidding algorithm.


%

%



\ifCLASSOPTIONcaptionsoff
  \newpage
\fi

\end{document}